\documentclass[11pt]{article}

\usepackage[utf8]{inputenc}

\usepackage[letterpaper,margin=1in]{geometry}
\usepackage{amsmath, amsthm}

\usepackage{amssymb}

\usepackage{complexity}

\usepackage{hyperref}
\hypersetup{colorlinks, linkcolor=darkblue, citecolor=darkgreen, urlcolor=darkblue}
\usepackage{times}

\usepackage[T1]{fontenc}
\usepackage{fix-cm}

\usepackage[inline]{enumitem}
\setlist[enumerate,1]{label=(\roman*), leftmargin=2.2em}
\setlist[enumerate,2]{label=(\alph*)}
\setlist{nosep,topsep=0.1em}
\setlist[itemize,1]{label={\bfseries--}}

\usepackage{mathtools}

\usepackage[linewidth=0.8pt]{mdframed}
\mdfsetup{skipabove=\bigskipamount, skipbelow=\bigskipamount, innerleftmargin=0.5em, innertopmargin=0.5em, innerrightmargin=0.5em, innerbottommargin=0.5em, align=center}

\usepackage{xcolor}
\definecolor{darkblue}{rgb}{0,0,0.38}
\definecolor{darkred}{rgb}{0.6,0,0}
\definecolor{darkgreen}{rgb}{0.1,0.35,0}
\definecolor{col1}{HTML}{31782F}
\colorlet{thmred}{red!90!brown!60!black}

\usepackage[
backend=biber,
style=alphabetic,
citestyle=alphabetic,
maxalphanames=6,
maxcitenames=99,
mincitenames=98,
maxbibnames=99,
giveninits=true,
url=false,
doi=true,
isbn=true,
backref=true
]{biblatex}

\usepackage{xpatch}

\makeatletter
\patchcmd\blx@bblinput{\blx@blxinit}
                      {\blx@blxinit
                      }{}{\fail}
\makeatother
\newcommand{\footremember}[2]{%
    \footnote{#2}
    \newcounter{#1}
    \setcounter{#1}{\value{footnote}}%
}
\newcommand{\footrecall}[1]{%
    \footnotemark[\value{#1}]%
}

\usepackage{xspace}

\makeatletter
 \newcommand{\linkdest}[1]{\Hy@raisedlink{\hypertarget{#1}{}}}
\makeatother
\newcommand{\OTSP}{\protect\hyperlink{prb:OTSP}{OTSP}\xspace}
\newcommand{\TSPPC}{\protect\hyperlink{prb:TSPPC}{TSP-PC}\xspace}
\newcommand{\TSP}{\protect\hyperlink{prb:TSP}{TSP}\xspace}

\newcommand{\Exp}{\mathbb{E}}
\newcommand{\Prob}{\mathbb{P}}
\newcommand{\odd}{\operatorname{odd}}
\newcommand{\e}{\mathrm{e}}
\newcommand{\PHK}{P_{\text{HK}}}

\newcommand{\Pstroll}[1][$s$-$t$]{P_\text{#1 stroll}}
\newcommand{\dotcup}{\mathbin{\dot{\cup}}}
\newcommand{\bigdotcup}{\mathbin{\dot{\bigcup}}}
\newcommand{\T}{\mathcal{T}}
\newcommand{\costLP}{c_{\text{LP}}}
\newcommand{\costOPT}{c_{\text{OPT}}}

\makeatletter
\DeclareFieldFormat{eprint:arxiv}{%
  arXiv\addcolon\space
  \ifhyperref
    {\href{https://arxiv.org/abs/#1}{%
       \nolinkurl{#1}%
       \iffieldundef{eprintclass}
     {}
     {\addspace\mkbibbrackets{\thefield{eprintclass}}}}}
    {\nolinkurl{#1}
     \iffieldundef{eprintclass}
       {}
       {\addspace\mkbibbrackets{\thefield{eprintclass}}}}}
\makeatother

\makeatletter
\newcommand{\labeltarget}[1]{\Hy@raisedlink{\hypertarget{#1}{}}}
\makeatother

\usepackage[margin=25pt,font=small,labelfont=bf]{caption}

\usepackage{subcaption}

\usepackage{graphicx}

\graphicspath{{.}{graphics/}}
\makeatletter
\newcommand\appendtographicspath[1]{%
  \g@addto@macro\Ginput@path{#1}%
}
\makeatother

\usepackage{tikz}
\usetikzlibrary{patterns}
\usetikzlibrary{decorations.pathmorphing}
\usetikzlibrary{calc}
\usetikzlibrary{backgrounds}
\usetikzlibrary{positioning}

\tikzset{terminal/.style={draw=black, fill=white, inner sep=1.3mm, very thick}}
\tikzset{inner/.style={circle, draw=black, fill=black, inner sep=.8mm, thick}}
\tikzset{edgeweight/.style={midway, rectangle, fill=white, inner sep=.8mm, opacity=0}}
\tikzset{full/.style={line width=2pt}}
\tikzset{three quarters/.style={line width=1.66pt}}
\tikzset{half/.style={loosely dashed, line width=1.33pt,}}
\tikzset{quarter/.style={loosely dotted, line width=1pt, line cap=butt}}
\tikzset{full node/.style={}}
\tikzset{three quarters node/.style={fill=black}}
\tikzset{half node/.style={pattern=north east lines}}
\tikzset{quarter node/.style={fill=white}}
\tikzset{edge/.style={thick, line cap=rect}}
\tikzstyle{join edge} = [draw, color=col1, dashed, ultra thick]
\tikzstyle{connector edge} = [draw, color=thmred, decorate, decoration={snake,amplitude=0.5mm,segment length=2mm}, ultra thick]
\tikzstyle{tree_edge} = [draw, black, ultra thick, line cap=rect]
\newcommand{\name}[1]{}

\makeatletter
\tikzset{%
   prefix node name/.code={%
      \tikzset{%
         name/.code={\edef\tikz@fig@name{#1 ##1}}
      }%
   }%
}
\makeatother

\usepackage[vlined,ruled,algo2e,linesnumbered]{algorithm2e}
\SetAlCapHSkip{0.2em}
\SetAlgoInsideSkip{smallskip}
\AtBeginDocument{%
  \SetCustomAlgoRuledWidth{0.81\textwidth}%
  \setlength{\algowidth}{0.975\textwidth}%
}
\makeatletter
\patchcmd{\@algocf@start}
  {\begin{lrbox}{\algocf@algobox}}
  {%
   \hspace*{0.095\textwidth}%
   \begin{lrbox}{\algocf@algobox}%
   \begin{minipage}{0.81\textwidth}%
  }
  {}{}
\patchcmd{\@algocf@finish}
  {\end{lrbox}}
  {\end{minipage}\end{lrbox}}
  {}{}
\makeatother

\usepackage[capitalize, nameinlink]{cleveref}
\crefname{theorem}{Theorem}{Theorems}
\Crefname{lemma}{Lemma}{Lemmas}
\Crefname{claim}{Claim}{Claims}
\Crefname{fact}{Fact}{Facts}
\Crefname{remark}{Remark}{Remarks}
\Crefname{observation}{Observation}{Observations}
\Crefname{figure}{Figure}{Figures}
\Crefname{algocf}{Algorithm}{Algorithms}
\Crefname{algoline}{Line}{Lines}

\newtheorem{theorem}{Theorem}
\newtheorem{lemma}[theorem]{Lemma}

\newtheorem{definition}[theorem]{Definition}
\newtheorem{remark}[theorem]{Remark}
\newtheorem{corollary}[theorem]{Corollary}

\newtheorem{observation}[theorem]{Observation}

\usepackage{xfrac}
\ExplSyntaxOn
\DeclareRestrictedTemplate { xfrac } { text } { math }
  {
    numerator-font      = \number \fam ,
    slash-symbol        = /            ,
    slash-symbol-font   = \number \fam ,
    denominator-font    = \number \fam ,
    scale-factor        = 0.7          ,
    scale-relative      = false        ,
    scaling             = true         ,
    denominator-bot-sep = 0 pt         ,
    math-mode           = true         ,
    phantom             = ( %
  }
\DeclareInstance { xfrac } { mathdefault } { math }
  { numerator-top-sep = 0pt }
\ExplSyntaxOff

\let\oldtop\top
\renewcommand{\top}{{\scriptscriptstyle{\oldtop}}}

\makeatletter
\def\@fnsymbol#1{\ensuremath{\ifcase#1\or *\or %
\ddagger\or
    \mathsection\or \mathparagraph\or \|\or **\or \dagger\dagger
    \or \ddagger\ddagger \else\@ctrerr\fi}}
\makeatother

\DeclareMathOperator*{\argmin}{arg\,min}

\title{A \texorpdfstring{$(\sfrac32+\sfrac1\e)$}{(3/2+1/e)}-Approximation Algorithm for Ordered TSP}

\author{%
Susanne Armbruster%
\footremember{Bonn}{%
Research Institute for Discrete Mathematics and Hausdorff Center for Mathematics, University of Bonn, Bonn, Germany.
Email:
\href{mailto:armbruster@or.uni-bonn.de}{armbruster@or.uni-bonn.de}, \href{mailto:mnaegele@uni-bonn.de}{mnaegele@uni-bonn.de}.
M.~N{\"a}gele is supported by the Swiss National Science Foundation (grant no.\ P500PT\_206742) and the Deutsche Forschungsgemeinschaft (DFG, German Research Foundation) under Germany's Excellence Strategy~--~EXZ-2047/1~--~390685813.
}%
\and
Matthias Mnich%
\footnote{%
Hamburg University of Technology, Institute for Algorithms and Complexity, Germany
Email:
\href{mailto:matthias.mnich@tuhh.de}{matthias.mnich@tuhh.de}.
Partially supported by the Deutsche Forschungsgemeinschaft (DFG, German Research Foundation), project MN 59/4-1.
}%
\and
Martin N{\"a}gele%
\footrecall{Bonn}%
}%
\date{}

\begin{document}

\maketitle

\begin{abstract}
We present a new $(\sfrac32+\sfrac1\e)$-approximation algorithm for the Ordered Traveling Salesperson Problem (Ordered TSP).
Ordered TSP is a variant of the classical metric Traveling Salesperson Problem (TSP) where a specified subset of vertices needs to appear on the output Hamiltonian cycle in a given order, and the task is to compute a cheapest such cycle.
Our approximation guarantee of approximately $1.868$ holds with respect to the value of a natural new linear programming (LP) relaxation for Ordered TSP.
Our result significantly improves upon the previously best known guarantee of~$\sfrac52$ for this problem and thereby considerably reduces the gap between approximability of Ordered TSP and metric TSP.
Our algorithm is based on a decomposition of the LP solution into weighted trees that serve as building blocks in our tour construction.
\end{abstract}

\thispagestyle{empty}
\addtocounter{page}{-1}

\begin{tikzpicture}[overlay, remember picture, shift = {(current page.south east)}]
\node[anchor=south east, outer sep=5mm] {
\begin{tikzpicture}[outer sep=0] %
\node (dfg) {\includegraphics[height=12mm]{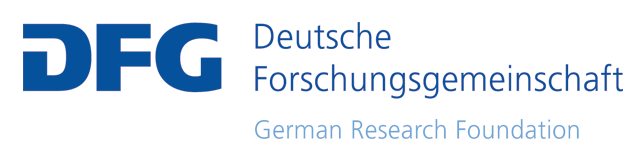}};
\node[left=5mm of dfg, yshift=1mm] (snf) {\includegraphics[height=8mm]{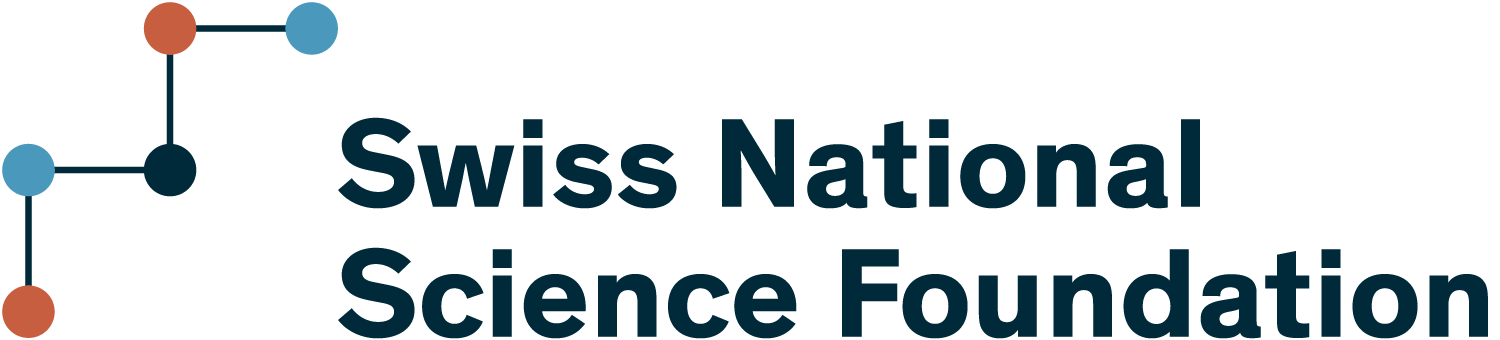}};
\end{tikzpicture}
};
\end{tikzpicture}

\newpage

\section{Introduction}

The classical metric Traveling Salesperson Problem (\linkdest{prb:TSP}\TSP) is one of the most fundamental and well-studied problems in Combinatorial Optimization and has a large number of applications.
A metric \TSP instance is given by a complete undirected graph~$G = (V, E)$ with metric edge cost~$c \colon E \to \mathbb{R}_{\geq 0}$.
The task is to find a cycle of minimum cost that visits each vertex exactly once, where the cost of a cycle equals the sum of the edge costs over all edges it contains.
Metric \TSP is highly relevant in many practical applications and thus, a lot of different variants are studied (see, e.g., \cite{SallerKK2023}).
The problem is \NP-hard and \APX-hard~\cite{PapadimitriouY1993}; concretely, assuming $\P\neq\NP$, it is known that no polynomial-time algorithm can guarantee to find a cycle of cost at most~$\sfrac{123}{122}$ times the cost of a cheapest cycle \cite{karpinski_2015_inapproximability}.
For a long time, the best-known approximation algorithm for metric \TSP was the Christofides-Serdyukov \sfrac32-approximation algorithm~\cite{christofides_1976_worst,christofides_1976_worst_new,serdyukov_1978_onekotorykh}.
This was recently improved to a breakthrough $(\sfrac32-\varepsilon)$-approximation algorithm, for some $\varepsilon > 10^{-36}$, by \textcite{karlin_2021_slightly,karlin_2023_derand}.

In this work, we focus on a generalization of metric \TSP known as \emph{Ordered TSP}, in which some of the vertices must be visited in a given order:

\begin{mdframed}[userdefinedwidth=0.95\linewidth]
\linkdest{prb:OTSP}{\textbf{Ordered TSP (\OTSP):}}
Given a complete undirected graph $G=(V,E)$ with metric edge cost $c\colon E\to\mathbb{R}_{\geq 0}$ and pairwise distinct vertices $d_1, \ldots, d_k\in V$, the task is to find a cheapest spanning cycle $C$ in $G$ that contains the vertices $d_1$, \ldots, $d_k$ in this order.
\end{mdframed}

\noindent
We typically refer to an input of \OTSP as an \emph{\OTSP instance} $(G, c, (d_1,\ldots,d_k))$; solutions are often called \emph{tours}.
Our goal in this paper is to further the understanding of the approximability of \OTSP, i.e., we aim to design $\alpha$-approximation algorithms for \OTSP with $\alpha$ as small as possible.

Clearly, \OTSP is at least as hard as metric \TSP, and therefore \APX-hard.
Surprisingly, not much more is known on the approximability of \OTSP.
\textcite{boeckenhauer_06_generalizations} observed that a $\sfrac52$-approximate solution can be readily obtained by first traversing $d_1,\ldots,d_k$ in this order and subsequently appending a tour on $V\setminus\{d_1,\ldots,d_k\}$ constructed through the Christofides-Serdyukov algorithm.
The black-box use of a metric \TSP approximation algorithm allows to reduce this guarantee by the same additive improvement of $\varepsilon > 10^{-36}$ as in the $(\sfrac32-\varepsilon)$-approximation by \citeauthor{karlin_2023_derand}.
Besides that, \textcite{boeckenhauer_2013_improved} gave a $(\sfrac52-\sfrac2k)$-approximation algorithm, where $k\geq 2$ is the number of ordered vertices in the \OTSP input.
Note that their result does not directly inherit the improvement achieved for metric \TSP, making its approximation ratio asymptotically inferior to the earlier approach of \citeauthor{boeckenhauer_06_generalizations}.
Finally, the intuition that \OTSP should become easier once~$k$ approaches $n$ is confirmed by a dynamic programming approach of \textcite{deineko_2006_inner} that runs in $O(2^{r}r^2n)$ time and $O(2^rrn)$ space, i.e., in polynomial time and space if $r\coloneqq n-k$, the number of vertices that are not in the input order, is of magnitude $O(\log n)$.

\bigskip

\OTSP is in fact a special case of a the following significantly more general \TSP variation termed \emph{TSP with Precedence Constraints}.

\begin{mdframed}[userdefinedwidth=0.95\linewidth]
\linkdest{prb:TSPPC}{\textbf{TSP with Precedence Constraints (\TSPPC):}}
Given a complete undirected graph $G=(V,E)$ with metric edge cost $c\colon E\to\mathbb{R}_{\geq 0}$ and a partial order $\prec$ on $V$, the task is to find a cheapest spanning cycle $C$ in $G$ that respects $\prec$, i.e., $C$ can be traversed such that whenever $u\prec v$ for two vertices $u,v\in V$, then $u$ appears earlier on $C$ than $v$.
\end{mdframed}

\noindent
Compared to the total order constraints in \OTSP, general partial orders allow for modeling a much wider range of problems.
One among many applications of \TSPPC is, e.g., tour planning for mixed pickup and delivery services, where one needs to make sure that a pickup happens before a delivery (but apart from that, pickups and deliveries can be intertwined arbitrarily).
There is a considerable body of research on the structure of the \TSPPC polyhedron, different dynamic programming algorithms, enhanced branch-and-bound methods, and various other exact and heuristic approaches, typically even for the more general version of \TSPPC with asymmetric edge cost (see, e.g., \cite{balas_1995_precedence,gouveia_2006_extended,salii_2019_revisiting,khachai_2023_precedence} and references therein).
Despite that, essentially no positive results on the approximability of \TSPPC are known, which is possibly explained by an influential hardness result of \textcite{CharikarMRS1997}:
By relating the problem to the \emph{Shortest Common Supersequence Problem}, they are able to show that there is no $(\log n)^{\delta}$-approximation for the path version of \TSPPC for any constant $\delta$, unless $\mathsf{NP}\subseteq \mathsf{DTIME}(n^{O(\log\log n)})$, even if the underlying metric space is a line.
This motivates our study of the approximability of \TSPPC on general metric spaces with special partial orders, i.e., \OTSP.

\subsection{Our results and techniques}

Our main contribution is to significantly improve the state of the art for \OTSP by giving an LP-relative approximation guarantee of $\sfrac32+\sfrac1\e\approx1.868$, as stated in the following theorem.

\begin{theorem}\label{thm:main}
There is a polynomial-time $\left(\sfrac32+\sfrac1\e\right)$-approximation algorithm for \OTSP.
\end{theorem}

This constitutes a significant improvement over the previous $(\sfrac52-\varepsilon)$-approximation algorithm.
We achieve this improvement by introducing a new linear programming (LP) relaxation for \OTSP and devising a suitable rounding procedure.
The LP relaxation is based on the Held-Karp relaxation that is typically leveraged in the context of \TSP, but allows for taking the prescribed order of the vertices $d_1,\ldots,d_k$ into account by using disjoint sets of variables to represent the $d_i$-$d_{i+1}$ strolls%
\footnote{%
We use the term \emph{$s$-$t$ stroll} instead of \emph{$s$-$t$ path} for a path from $s$ to $t$ in the underlying graph to emphasize that we do not require all vertices to be covered.
Also, for convenience of notation, we use $d_{k+1}\coloneqq d_1$ throughout the paper.
}
that a solution is composed of.
Our rounding procedure crucially relies on a result on decomposing (fractional) $s$-$t$ strolls into a convex combination of trees.
This decomposition resembles an existential result by \textcite[Theorem~2.6]{bang-jensen_1995_preserving} on packing branchings in a directed multigraph.
Variations thereof have recently been used for advances on another variant of \TSP, namely \emph{Prize-Collecting TSP} \cite{blauth_2023_improved,blauth_2024_better}, and motivate the application here. (See \cref{lem:decomposition} for the precise statement of the decomposition result.)
The trees obtained from stroll decompositions enable the construction of a subgraph that spans a reasonably large part of~$V$ at cost no more than the LP solution cost, and contains a walk with visits at $d_1,\ldots,d_k$ in this order.
Our tour construction is completed by connecting the remaining isolated vertices in a cheapest possible way, and applying a parity correction step as typical for \TSP-like problems.

\bigskip

Our approach crucially relies on being able to split a solution into $d_i$-$d_{i+1}$ strolls upfront, hence it is not directly suitable for handling arbitrary precedence constraints other than total orders.
While one can always try to guess a suitable total order that is compatible with the given partial order, and then apply \cref{thm:main}, this is generally not efficient.
We can, though, obtain approximation algorithms for some special cases of precedence constraints, as for example in the following result that is a direct generalization of \cref{thm:main}.

\begin{theorem}\label{thm:indep_orders}
Consider a \TSPPC instance $(G, c, \prec)$ on a complete graph $G=(V,E)$ with a partial order~$\prec$ that can be equivalently given as independent total orders on disjoint subsets $D_1,\ldots,D_\ell\subseteq V$.
There is a polynomial-time $(\ell+\sfrac12+\sfrac1{\e^\ell})$-approximation algorithm for this class of \TSPPC problems.
\end{theorem}

The total orders on the sets~$D_i$ are also called \emph{chains}.
We remark that losing a factor of $\ell$ in \cref{thm:indep_orders} is intrinsic to our approach: We never merge the given chains, but traverse them one after another.
Still, the result of \cref{thm:indep_orders} is superior to a black-box algorithm that independently applies the algorithm from \cref{thm:main} to the $\ell$ chains and concatenates the resulting tours (while shortcutting to avoid repeated visits).

\subsection{Related Work}
Variations of \OTSP and \TSPPC are also studied in the context of scheduling with precedence constraints.
In a classical setup, denoted by $Pm|\textnormal{prec}|C_{\max}$ in the scheduling literature, one needs to find a schedule for a set $\mathcal J$ of $n$ jobs on $m$ identical machines subject to precedence constraints between the jobs.
Formally, each job $j\in \mathcal J$ is characterized by a processing time $p_j\in\mathbb Z_{\geq 0}$, and a schedule $\sigma\colon\mathcal J\rightarrow\mathbb Z_{\geq 0}\times\{1,\ldots,m\}$ assigns each job $j\in \mathcal J$ to a pair $(\sigma_1(j), \sigma_2(j))$ consisting of an integer start time $\sigma_1(j)$ and a machine $\sigma_2(j)$ such that no other job scheduled on that machine has their start time in the time interval $[\sigma_1(j),\sigma_1(j)+p_j]$, and for any two jobs $j,j'\in\mathcal J$ related as $j \prec j'$ it holds that $\sigma_1(j) < \sigma_2(j')$.
The makespan objective~$C_{\max}$ of a schedule $\sigma$ is the maximum completion time $C_j = \sigma_1(j)+p_j$ over all jobs $j\in \mathcal J$.
Generally, precedence constraints of this type are studied extensively in a wide range of scheduling problems, including different settings and objectives (see, e.g., \cite{svensson_proc_scheduling, graham_scheduling, rothvoss_scheduling, levey_rothvoss_m_3_scheduling}).
The three-machine problem~$P3 | \textnormal{prec}, p_j\equiv 1 | C_{\max}$ is one of the few famous open problems by \textcite{garey_1979_computers} whose computational complexity has not yet been resolved.

The complexity of many scheduling problems with precedence constraints that are chains has been well-investigated.
An influential paper of \textcite{LenstraRinnooyKan1980} shows strong $\mathsf{NP}$-hardness for minimizing the number of chain-constrained unit-size jobs that miss their deadline on a single machine.
Several other works with chain constraints have appeared~\cite{JansenSolisOba2010,DuLY1991,Kunde1981,Woeginger2000}.

Towards analogues of \TSPPC, we may consider the aforementioned problem $Pm|\textnormal{prec}|C_{\max}$ on a single machine, but add sequence-dependent setup times $s_{ij}\in\mathbb Z_{\geq0}$ between any two jobs $i$ and $j$, which add to the makespan of the schedule.
This problem, which is denoted as $1|\textnormal{prec},s_{ij}|C_{\max}$, was discussed by \textcite[Section~3.1.2]{LiaeeEmmons1996}.
In case of \TSPPC, the setup times are metric (i.e., $s_{ij} \leq s_{ik} + s_{kj}$ for any triple $(i,j,k)$ of distinct jobs), and all jobs have equal processing time $p_j\equiv 0$.
To be precise, the objective function for \TSPPC takes into account the cost for returning to the origin city whereas no such cost occurs in the objective function for $1|\textnormal{prec},s_{ij}|C_{\max}$, hence the latter in fact models a path version of \TSPPC.

\subsection{Organization of the paper}

In \cref{sec:algorithm}, we introduce our new linear programming formulation for \OTSP (\cref{subsec:polyhedra}) and analyze a randomized algorithm giving the guarantee of \cref{thm:main} in expectation (\cref{subsec:rounding}).
We show how this algorithm can be derandomized in \cref{subsec:derandomization}.
Finally, \cref{sec:indep_orders} extends our framework to yield \cref{thm:indep_orders}, and \cref{sec:decomposition_proof} shows how our main technical lemma is implied by a closely related known result.

\section{Our algorithm}\label{sec:algorithm}

\subsection{The LP relaxation and polyhedral basics}\label{subsec:polyhedra}

The most commonly used LP relaxation in approximation algorithms for classical \TSP is the so-called \emph{Held-Karp relaxation}.
It was first introduced by \textcite{dantzig_fulkerson_johnson} and is given by
$$
\PHK(G) \coloneqq \left\{x\in\mathbb{R}_{\geq 0}^E\colon
\begin{array}{rcll}
x(\delta(v)) & = & 2 & \forall v\in V\\
x(\delta(S)) & \geq & 2 & \forall S\subsetneq V, S\neq\emptyset
\end{array}\right\}\enspace,
$$
where $G=(V,E)$ is the underlying complete graph.%
\footnote{For $S \subseteq V$ we denote by~$\delta(S)$ the set of edges with exactly one endpoint in~$S$.
For~$v \in V$, we abbreviate~$\delta(v) \coloneqq \delta(\{v\})$.}
While \TSP simply asks for a spanning cycle, \OTSP requires that the vertices $d_1$, \ldots, $d_k$ appear on the cycle in this order.
Thus, a solution is naturally composed of $k$ strolls, namely a $d_i$-$d_{i+1}$ stroll for every $i\in\{1,\ldots,k\}$.
For a polyhedral description of $s$-$t$ strolls in a complete graph $G=(V,E)$, we modify the Held-Karp relaxation for \emph{$s$-$t$ path TSP}%
\footnote{Given a complete graph $G=(V,E)$ with metric edge costs and vertices $s,t\in V$, \emph{$s$-$t$ path TSP} is the variant of \TSP that seeks a path of smallest total cost from $s$ to~$t$ while visiting every vertex exactly once.%
}
to allow partial coverage of vertices.
Concretely, the variables $y\in\mathbb{R}_{\geq 0}^V$ in the following formulation indicate the extent at which vertices are covered:%
\footnote{%
\label{fn:stroll_coverage_variables}%
The constraints of $\Pstroll$ imply that for $v\in V\setminus \{s,t\}$, we have $2 y_v \leq x(\delta(V\setminus\{s,t\})) \leq x(\delta(s)) + x(\delta(t)) \leq 2$, and thus~$y_v\leq 1$, legitimating the proposed interpretation.
}
\begin{equation}
\label{eqn:polyhedral_stroll}
\Pstroll(G)\coloneqq \left\{
(x, y)\in\mathbb{R}_{\geq 0}^E\times\mathbb{R}_{\geq 0}^V\colon
\begin{array}{rcll}
x(\delta(v)) & = & 2 y_v & \forall v\in V\\
x(\delta(S)) & \geq & 1 & \forall S\subseteq V\setminus \{t\}, s\in S\\
x(\delta(S)) & \geq & 2 y_v & \forall S\subseteq V\setminus \{s,t\}, v\in S\\\
y_s = y_t & = & \sfrac12
\end{array}
\right\} \enspace .
\end{equation}
Note that setting $y_s=y_t=\sfrac12$ corresponds to $s$ and $t$ having degree $1$ in an $s$-$t$ stroll, while all interior vertices of an integral stroll have degree $2$, which corresponds to a $y$-value of $1$.
Using the above polyhedral relaxation~\eqref{eqn:polyhedral_stroll} for all $d_i$-$d_{i+1}$ strolls, it remains to link the strolls by requiring full joint coverage of every $v\in V$.
This results in the following LP relaxation for \OTSP:
\begin{align*}\label{eq:lp_relaxation}\tag{\OTSP LP relaxation}
\begin{array}{rrcll}
\min & \displaystyle \sum_{e\in E} c_e \sum_{i=1}^k x^i_e \\
     & \displaystyle\sum_{i=1}^k y^i_v & = & 1 & \forall v\in V\\
     & (x^i, y^i) & \in & \Pstroll[$d_i$-$d_{i+1}$](G) & \forall i\in\{1,\ldots k\}\enspace.
\end{array}
\end{align*}

It is clear that any \OTSP solution can be turned into a feasible solution to the above LP of the same objective value, hence the above LP is indeed a relaxation of \OTSP.
We first observe that this \ref{eq:lp_relaxation} strengthens the Held-Karp relaxation in the following sense.

\begin{observation}\label{obs:HK_feasibility}
Let $(x^i,y^i)_{i\in\{1,\ldots,k\}}$ be feasible for the \ref{eq:lp_relaxation}.
Then $x\coloneqq\sum_{i=1}^k x^i\in\PHK(G)$.
\end{observation}

\begin{proof}
To see that $x$ satisfies the degree constraints in $\PHK$, note that for all $v\in V$, we have
$$
x(\delta(v)) = \sum_{i=1}^k x^i(\delta(v)) = 2 \cdot \sum_{i=1}^k y_i = 2\enspace.
$$
To verify the cut constraints, let $S\subsetneq V$ be a non-empty set of vertices.
If both $S\cap\{d_1,\ldots, d_k\}\neq\emptyset$ and~$(V\setminus S)\cap\{d_1,\ldots,d_k\}\neq\emptyset$, then there exist two distinct indices $i_1,i_2\in\{1,\ldots,k\}$ such that $d_{i_1}\in S$ but $d_{i_1+1}\notin S$, and $d_{i_2}\notin S$ but $d_{i_2+1}\in S$.
This implies that $x^{i_1}(\delta(S))\geq 1$ and $x^{i_2}(\delta(V\setminus S))\geq 1$, so we get
$$
x(\delta(S)) = \sum_{i=1}^k x^i(\delta(S)) \geq x^{i_1}(\delta(S)) + x^{i_2}(\delta(S)) = x^{i_1}(\delta(S)) + x^{i_2}(\delta(V\setminus S)) \geq 2\enspace.
$$
Otherwise, assume without loss of generality that $S\cap\{d_1,\ldots,d_k\}$ is empty (if not, $V\setminus S$ has this property) and fix a vertex $v\in S$.
We then know that $x^i(\delta(S))\geq 2y_v$ for all $i\in\{1,\ldots,k\}$, hence
\begin{equation*}
x(\delta(S)) = \sum_{i=1}^k x^i(\delta(S)) \geq 2\cdot\sum_{i=1}^k y_v = 2\enspace.\qedhere
\end{equation*}
\end{proof}

The point $x\in\PHK(G)$ constructed in \cref{obs:HK_feasibility} has the property that its cost $c^\top x$ equals the objective value $\costLP$ of the feasible point of the \ref{eq:lp_relaxation} that we started with.
Thus, following the arguments of \citeauthor{wolsey_1980_heuristic}'s polyhedral analysis \cite{wolsey_1980_heuristic} of the Christofides-Serdyukov algorithm, we immediately obtain the following.

\begin{corollary}\label{cor:tree_join_cost}
Let $\costLP$ denote the optimal objective value of the \ref{eq:lp_relaxation}.
Then, in the underlying graph $G$ with edge costs $c$, the following holds true.
\begin{enumerate}
\item\label{item:tree} A shortest spanning tree $T$ satisfies $c(T)\leq \costLP$.
\item\label{item:join} For any even cardinality set $Q\subseteq V$, a shortest $Q$-join $J$ satisfies $c(J)\leq \frac12\cdot \costLP$.
\end{enumerate}
\end{corollary}

\begin{proof}
Let $(x^i,y^i)_{i\in\{1,\ldots,k\}}$ be an optimal solution of the \ref{eq:lp_relaxation}.
By \cref{obs:HK_feasibility}, $x\coloneqq\sum_{i=1}^k x^i\in\PHK(G)$.
It is well-known due to \textcite{held_karp_spanning_trees} that then, $\frac{|V|-1}{|V|}\cdot x$ is feasible for the spanning tree polytope, and due to \textcite{wolsey_1980_heuristic} that $\frac12 x$ is feasible for the dominant of the $Q$-join polytope, hence $c(T)\leq \smash{\frac{|V|-1}{|V|}}\cdot c^\top x < c^\top x$ and $c(J)\leq \frac12 c^\top x$.
Using that $c^\top x = \costLP$, the result follows.
\end{proof}

\subsection{Rounding an LP solution}\label{subsec:rounding}

At its core, our algorithm for rounding a typically fractional solution $(x^i, y^i)_{i\in\{1,\ldots,k\}}$ of the \ref{eq:lp_relaxation} is based on leveraging a decomposition result for each of the points $(x^i, y^i)\in\Pstroll[$d_i$-$d_{i+1}$]$.
By scaling up $(x^i, y^i)$ by a large enough factor $M$ such that $M x^i$ is integral, this decomposition can be viewed as a result on packing trees into the multigraph that has $M x_e^i$ copies of every edge $e\in E$ and such that every vertex $v$ appears in $M y^i$ many of the trees.
While most results of this type deal with packing spanning trees (or, in the directed case, arborescences), i.e., consider uniform packings, \textcite{bang-jensen_1995_preserving} gave one of few results in a non-uniform setting as we are facing here.
Their splitting-off based construction was revised by \textcite{blauth_2023_improved} to obtain more fine-grained control over the output components of the decomposition when starting from a solution of a Held-Karp-type relaxation that allows partial coverage of vertices (similar to what we allow in $\Pstroll$).
We observe that these findings can be immediately carried over to solutions of $\Pstroll$, giving \cref{lem:decomposition} below.
We defer a formal proof to \cref{sec:decomposition_proof}.

\begin{lemma}\label{lem:decomposition}
Let $G=(V,E)$ be an undirected graph, $s,t\in V$, and let $(x,y)\in\Pstroll(G)$.
We can in polynomial time compute a family $\T$ of subtrees of $G$ that all contain the vertices $s$ and $t$, and weights $\mu\in[0,1]^{\T}$ with $\sum_{T\in\T}\mu_T=1$ such that%
\footnote{For a graph~$H$ we denote by~$V[H]$ the set of vertices and by~$E[H]$ the set of edges of~$H$.}
\begin{equation*}
\sum_{T\in\T} \mu_T\chi^{E[T]} = x
\qquad\text{and}\qquad
\sum_{T\in\T\colon v \in V[T]} \mu_T = y_v \quad \forall v\in V\setminus\{s,t\}\enspace.
\end{equation*}
\end{lemma}

In other words, \cref{lem:decomposition} allows to decompose a fractional $s$-$t$ stroll into a convex combination of trees in a family $\T$ that all connect $s$ and $t$, and such that for every other vertex $v\in V\setminus \{s,t\}$, the weighted number of trees that contain $v$ equals the coverage $y_v$ of $v$ in the stroll.
An example of a feasible solution $(x,y)$ and a decomposition satisfying the properties of \cref{lem:decomposition} is given in \cref{fig:decomposition}.

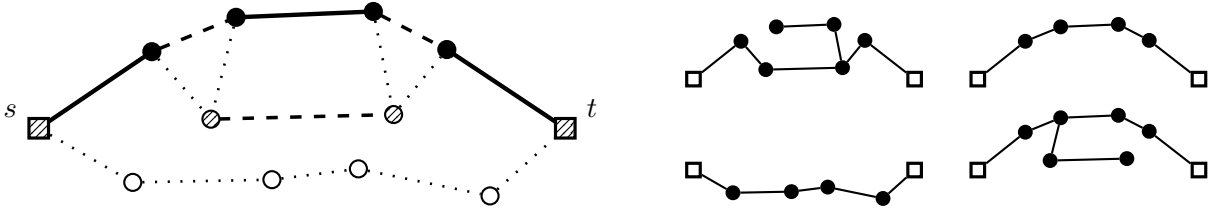
\begin{figure}[ht]
\centering
\captionsetup[subfigure]{position=b}
\subcaptionbox{Solution $(x,y)$ with $x_e=\sfrac14$ for dotted edges, \mbox{$x_e=\sfrac12$} for dashed edges, and $x_e=\sfrac34$ for solid edges. Likewise, \mbox{$y_v=\sfrac14$} for blank vertices, $y_v=\sfrac12$ for dashed vertices, and $y_v=\sfrac34$ for full vertices.%
\label{subfig:stroll_sol}}{%
\begin{tikzpicture}[scale=0.5, yscale=.72]
   \node[opacity=1] at (-0.5, 1) {};
   \node[opacity=1] at (16.5,1) {};
   \node[terminal, half node, label=175:$s$] at (1, 1) (s) {};
   \node[terminal, half node, label=5:$t$] at (15, 1) (t) {};

   \pgfmathsetseed{3}
   \node[inner, three quarters node] at ($(4, 3.5)+(0.5*rand, 0.5*rand)$) (v1) {\name{v_1}};
   \node[inner, three quarters node] at ($(6, 5.5)+(0.5*rand, 0.5*rand)$) (v2) {\name{v_2}};
   \node[inner, half node] at ($(6, 1.5)+(0.5*rand, 0.5*rand)$) (v3) {\name{v_3}};
   \node[inner, three quarters node] at ($(10, 5.5)+(0.5*rand, 0.5*rand)$) (v4) {\name{v_4}};
   \node[inner, half node] at ($(10, 1.5)+(0.5*rand, 0.5*rand)$) (v5) {\name{v_5}};
   \node[inner, three quarters node] at ($(12, 3.5)+(0.5*rand, 0.5*rand)$) (v6) {\name{v_6}};
   \node[inner, quarter node] at (3.5, -1) (v7) {\name{v_7}};
   \node[inner, quarter node] at (7.2, -0.9) (v8) {\name{v_8}};
   \node[inner, quarter node] at (9.5, -0.5) (v9) {\name{v_9}};
   \node[inner, quarter node] at (13, -1.5) (v10) {\name{v_{10}}};

   \draw[edge, three quarters] (s) -- (v1) node [edgeweight] (weight1) {$\frac 34$};
   \draw[edge, half] (v1) -- (v2) node [edgeweight] (weight2) {$\frac 12$};
   \draw[edge, quarter] (v1) -- (v3) node [edgeweight] (weight3) {$\frac 14$};
   \draw[edge, quarter] (v2) -- (v3) node [edgeweight] (weight4) {$\frac 14$};
   \draw[edge, three quarters] (v2) -- (v4) node [edgeweight] (weight5) {$\frac 34$};
   \draw[edge, half] (v3) -- (v5) node [edgeweight] (weight6) {$\frac 12$};
   \draw[edge, quarter] (v4) -- (v5) node [edgeweight] (weight7) {$\frac 14$};
   \draw[edge, half] (v4) -- (v6) node [edgeweight] (weight8) {$\frac 12$};
   \draw[edge, quarter] (v5) -- (v6) node [edgeweight] (weight9) {$\frac 14$};
   \draw[edge, three quarters] (v6) -- (t) node [edgeweight] (weight10) {$\frac 34$};

   \draw[edge, quarter] (s) -- (v7) node [edgeweight] (weight11) {$\frac 14$};
   \draw[edge, quarter] (v7) -- (v8) node [edgeweight] (weight12) {$\frac 14$};
   \draw[edge, quarter] (v8) -- (v9)node [edgeweight] (weight13) {$\frac 14$};
   \draw[edge, quarter] (v9) -- (v10) node [edgeweight] (weight14) {$\frac 14$};
   \draw[edge, quarter] (v10) -- (t) node [edgeweight] (weight15) {$\frac 14$};
\end{tikzpicture} %
}
\hfill
\subcaptionbox{A decomposition of the solution $(x,y)$ given in (\subref{subfig:stroll_sol}) into four trees with uniform weight $\mu\equiv\sfrac14$ satisfying the properties of \cref{lem:decomposition}.}{%
\begin{tikzpicture}[every path/.append style={thick}, scale=0.21, yscale=0.72, terminal/.append style={inner sep=.9mm}, inner/.append style={inner sep=.6mm}]

\begin{scope}[prefix node name=T1]
   \footnotesize
   \node[terminal] at (1, 1) (s){};
   \node[terminal] at (15, 1) (t){};

   \node[inner] at (3.5, -1) (v7) {\name{v_7}};
   \node[inner] at (7.2, -0.9) (v8) {\name{v_8}};
   \node[inner] at (9.5, -0.5) (v9) {\name{v_9}};
   \node[inner] at (13, -1.5) (v10) {\name{v_{10}}};

   \draw (T1 s) -- (T1 v7) -- (T1 v8) -- (T1 v9) -- (T1 v10) -- (T1 t);

\end{scope}

\begin{scope}[shift ={(0, 8)}, prefix node name=T2]
   \footnotesize
   \node[terminal] at (1, 1) (s){};
   \node[terminal] at (15, 1) (t){};

   \pgfmathsetseed{3}
   \node[inner] at ($(4, 4)+(0.5*rand, 0.5*rand)$) (v1) {\name{v_1}};
   \node[inner] at ($(6, 6)+(0.5*rand, 0.5*rand)$) (v2) {\name{v_2}};
   \node[inner] at ($(6, 2)+(0.5*rand, 0.5*rand)$) (v3) {\name{v_3}};
   \node[inner] at ($(10, 6)+(0.5*rand, 0.5*rand)$) (v4) {\name{v_4}};
   \node[inner] at ($(10, 2)+(0.5*rand, 0.5*rand)$) (v5) {\name{v_5}};
   \node[inner] at ($(12, 4)+(0.5*rand, 0.5*rand)$) (v6) {\name{v_6}};

   \draw (T2 s) -- (T2 v1) -- (T2 v3) -- (T2 v5) -- (T2 v6) -- (T2 t);
   \draw (T2 v5) -- (T2 v4) -- (T2 v2);
\end{scope}

\begin{scope}[shift = {(18,0)}, prefix node name=T3]

   \footnotesize
   \node[terminal] at (1, 1) (s){};
   \node[terminal] at (15, 1) (t){};

   \pgfmathsetseed{3}
   \node[inner] at ($(4, 4)+(0.5*rand, 0.5*rand)$) (v1) {\name{v_1}};
   \node[inner] at ($(6, 6)+(0.5*rand, 0.5*rand)$) (v2) {\name{v_2}};
   \node[inner] at ($(6, 2)+(0.5*rand, 0.5*rand)$) (v3) {\name{v_3}};
   \node[inner] at ($(10, 6)+(0.5*rand, 0.5*rand)$) (v4) {\name{v_4}};
   \node[inner] at ($(10, 2)+(0.5*rand, 0.5*rand)$) (v5) {\name{v_5}};
   \node[inner] at ($(12, 4)+(0.5*rand, 0.5*rand)$) (v6) {\name{v_6}};

   \draw (T3 s) -- (T3 v1) -- (T3 v2) -- (T3 v4) -- (T3 v6) -- (T3 t);
   \draw (T3 v2) -- (T3 v3) -- (T3 v5);

\end{scope}

\begin{scope}[shift={(18,8)}, prefix node name=T4]
   \footnotesize
   \node[terminal] at (1, 1) (s){};
   \node[terminal] at (15, 1) (t){};

   \pgfmathsetseed{3}
   \node[inner] at ($(4, 4)+(0.5*rand, 0.5*rand)$) (v1) {\name{v_1}};
   \node[inner] at ($(6, 6)+(0.5*rand, 0.5*rand)$) (v2) {\name{v_2}};
   \node[inner, opacity=0] at ($(6, 2)+(0.5*rand, 0.5*rand)$) (v3) {\name{v_3}};
   \node[inner] at ($(10, 6)+(0.5*rand, 0.5*rand)$) (v4) {\name{v_4}};
   \node[inner, opacity=0] at ($(10, 2)+(0.5*rand, 0.5*rand)$) (v5) {\name{v_5}};
   \node[inner] at ($(12, 4)+(0.5*rand, 0.5*rand)$) (v6) {\name{v_6}};

   \draw (T4 s) -- (T4 v1) -- (T4 v2) -- (T4 v4) -- (T4 v6) -- (T4 t);

\end{scope}

\end{tikzpicture} %
}
\caption{A solution $(x,y)\in\Pstroll$ along with a decomposition into trees, exemplifying \cref{lem:decomposition}.}
\label{fig:decomposition}
\end{figure}

After applying \cref{lem:decomposition} to all strolls $(x^i, y^i)\in\Pstroll[$d_i$-$d_{i+1}$]$ obtained from an optimal solution of the \ref{eq:lp_relaxation}, we choose one tree from each of the decompositions and consider the (\mbox{multi}\nobreakdash-)union of all edges obtained this way.
This results in a graph that already contains a closed walk with visits at~$d_1,\ldots,d_k$ in this order, giving the basis for our construction of an \OTSP solution.
Also, we can easily bound the expected cost of the edge set obtained in this way by randomly choosing the trees with marginals given by the weights~$\mu$ from \cref{lem:decomposition}.
To obtain an actual \OTSP solution, the missing steps are to \begin{enumerate*}
\item connect vertices that are not covered by any of the trees $T_i$,
\item perform parity correction to guarantee that there exists an Eulerian tour, and
\item shortcut appropriately to obtain an actual \OTSP solution.
\end{enumerate*}
Altogether, this leads to the randomized \cref{alg:ordered_tsp} as laid out below; also see \cref{fig:main_example} for an example illustration of the different edge sets that are constructed in \cref{alg:ordered_tsp}.

\begin{algorithm2e}[!b]
   \KwIn{\OTSP instance $(G, c, \{d_1,\ldots,d_k\})$ on graph $G=(V,E)$.}
   \medskip

   \DontPrintSemicolon
   Compute an optimal solution~$(x^i, y^i)_{i\in\{1,\ldots,k\}}$ to the \ref{eq:lp_relaxation}.\;
   \ForEach{$i \in \{1, \dots, k\}$}
   {
      Apply \cref{lem:decomposition} to decompose~$(x^i,y^i)$ into trees~$\T_i$ with weights $\mu^i$. \label{algline:decomposition}\;
      Sample one tree~$T_i$ from $\T_i$ with marginals given by~$\mu^i$.\label{algline:sample_T}\;
   }
   Compute a minimum-cost edge set $F\subseteq E$ such that the multigraph
   $$H \coloneqq \left(V, F\cup\bigdotcup_{i \in \{1, \dots, k\}} E[T_i]\right)$$
   is connected.\label{algline:connector_F}\;
   Let~$Q = \operatorname{odd}(H)$ and compute a minimum cost~$Q$-join~$J$ in~$G$. \label{algline:join_computation}\;
   \Return Spanning cycle $C$ in $G$ obtained from~$H\dotcup J$ through \cref{lem:spanning_cycle}.\label{algline:H_cup_J}\;
   \caption{A randomized approximation algorithm for \OTSP}
   \label{alg:ordered_tsp}
\end{algorithm2e}

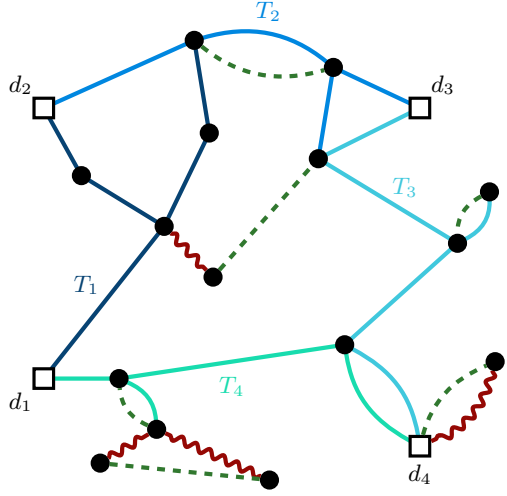
\begin{figure}
\centering
\begin{tikzpicture}[scale=-0.5, yscale=0.9]

\definecolor{T1_col}{HTML}{074373}
\definecolor{T2_col}{HTML}{0087e0}
\definecolor{T3_col}{HTML}{40c8dd}
\definecolor{T4_col}{HTML}{18dcae}

   \footnotesize
   \begin{scope}[every node/.style={inner}]
      \node (1) at (2.7,1.5) {};
      \node (2) at (2.3,-1.2) {};
      \node (3) at (-1,4) {};
      \node (4) at (2,7) {};
      \node (5) at (5.5,5) {};
      \node (6) at (5.6,0.74) {};
      \node (7) at (4,11) {};
      \node (8) at (6.8,3.5) {};
      \node (9) at (9,2) {};
      \node (11) at (8,8) {};
      \node (12) at (6,-2) {};
      \node (13) at (-1.85,2.48) {};
      \node (14) at (-2,7.5) {};
      \node (15) at (7,9.5) {};
      \node (16) at (8.5,10.5) {};
   \end{scope}

   \begin{scope}[every node/.style={terminal}]
      \node[label={[circle, inner sep=0pt]45:$d_3$}] (d3) at (0,0) {};
      \node[label={[circle, inner sep=0pt]-90:$d_4$}] (d4) at (0,10) {};
      \node[label={[circle, inner sep=0pt]-135:$d_1$}] (d1) at (10, 8) {};
      \node[label={[circle, inner sep=0pt]135:$d_2$}] (d2) at (10,0) {};
   \end{scope}

   \begin{scope}[on background layer]
	\begin{scope}[every path/.style={tree_edge, T3_col}]
      \draw (d3) -- (1);
      \draw (1) -- node[above right = -3pt] (text) {$T_3$} (3);
      \draw (3) -- (4);
      \draw (4) to [bend left=25] (d4);
      \draw (3) to [bend right] (13);
	\end{scope}
	\begin{scope}[every path/.style={tree_edge, T4_col}]
      \draw (d4) to [bend left=25] (4);
      \draw (4) -- node[below = 2pt] (text) {$T_4$}(11);
      \draw (11)  -- (d1);
      \draw (11) to [bend left] (15);
	\end{scope}
	\begin{scope}[every path/.style={tree_edge, T1_col}]
      \draw (d1) -- node [above left=-2pt] (text) {$T_1$} (8);
      \draw (8) -- (9);
      \draw (9)  -- (d2);
      \draw (8) -- (6);
      \draw (6) -- (12);
	\end{scope}
	\begin{scope}[every path/.style={tree_edge, T2_col}]
      \draw (d2) -- (12);
      \draw (12) to[bend left] node [above = 1pt] (text) {$T_2$} (2);
      \draw (2) -- (d3);
      \draw (2) -- (1);
	\end{scope}
   \end{scope}

   \begin{scope}[every path/.style={connector edge}]
      \draw (14) to [bend left=25] coordinate (example_connector) (d4);
      \draw (16) -- (15);
      \draw (7) -- (15);
      \draw (8) to (5);
   \end{scope}

   \begin{scope}[every path/.style={join edge}]
      \draw (16) to coordinate (example_join) (7);
      \draw (d4) to [bend left=25] (14);
      \draw (15) to [bend left] (11);
      \draw (12) to [bend right] (2);
      \draw (5) to (1);
      \draw (13) to [bend right] (3);
   \end{scope}

\end{tikzpicture} \caption{Exemplifying the construction of Eulerian graph $H\dotcup J$ from \cref{alg:ordered_tsp}: Trees $T_1,T_2,T_3,T_4$ drawn as solid blueish edges, the edge set $F$ connecting all vertices to the trees drawn as curly red edges, and the $\odd(H)$-join $J$ drawn as dashed green edges.}
\label{fig:main_example}
\end{figure}

We first show that \cref{alg:ordered_tsp} gives the guarantees claimed by \cref{thm:main} in expectation and---even stronger---with respect to the value $\costLP$ of the \ref{eq:lp_relaxation}, as stated in the subsequent theorem.
In \cref{subsec:derandomization}, we show that \cref{alg:ordered_tsp} admits an immediate derandomization using the method of conditional expectation, thereby completing the proof of \cref{thm:main}.

\begin{theorem}\label{thm:main_randomized}
Let~$\costLP$ be the cost of an optimum solution of the \ref{eq:lp_relaxation}.
\cref{alg:ordered_tsp} returns in polynomial time an \OTSP solution $C$ satisfying
$$
\Exp[c(E[C])] \leq \left(\frac32 + \frac1\e\right) \cdot \costLP\enspace.
$$
\end{theorem}

To prove \cref{thm:main_randomized}, we first study the random graph $H_0\coloneqq (V, \bigdotcup_{i \in \{1, \dots, k\}}E[T_i])$ obtained from taking the union of trees $T_i\in\T_i$ for all $i\in\{1,\ldots,k\}$ as sampled in \cref{alg:ordered_tsp}.
In order for the following statements to also be applicable in a proof of \cref{thm:indep_orders}, we refer to the tree distributions of the type generated in \cref{alg:ordered_tsp} as \emph{connecting tree distributions}.
\begin{definition}[Connecting tree distribution]
   Let~$G = (V, E)$ be a graph and let~$d_1, \dots, d_k \in V$.
   A connecting tree distribution~$(\T_i, \mu^i)_{i \in \{1, \dots, k\}}$ consists of a family~$\T_i$ of subtrees of~$G$ and marginals~$\mu^i \colon \T_i \to (0, 1]$ for every $i\in\{1,\ldots,k\}$ with the following properties.
   \begin{enumerate}
      \item $\sum_{T \in \T_i} \mu^i_T = 1$ for all~$i \in \{1, \dots, k\}$.
      \item $V[T] \cap \{d_1, \dots, d_k\} = \{d_i, d_{i + 1}\}$ for all~$T \in \T_i$ and~$i \in \{1, \dots, k\}$.
      \item \label{defitem:joint_coverage} $\sum_{i = 1}^k \sum_{T \in \T_i \colon v \in V[T]} \mu^i_T = 1$ for all $v\in V\setminus\{d_1,\ldots,d_k\}$.
   \end{enumerate}
\end{definition}

The distributions $(\T_i, \mu^i)$ obtained in \cref{alg:ordered_tsp} by applying \cref{lem:decomposition} indeed satisfy the constraints of the above definition; in particular, \cref{defitem:joint_coverage} is fulfilled because
\begin{equation*}
   \sum_{i=1}^k\sum_{T\in\T_i\colon v \in V[T]} \mu^i_T = \sum_{i=1}^k y_v^i = 1 \quad \forall v\in V\setminus\{d_1,\ldots,d_k\}\enspace,
\end{equation*}
where the first equality follows from \cref{lem:decomposition}, and the second one is implied by constraints of $\Pstroll[$d_i$-$d_{i+1}$]$.

\begin{lemma}\label{lem:tree_cost}
Let~$G = (V, E)$ be a graph,~$d_1, \dots, d_k \in V$, and let $(\T_i, \mu^i)$ be a connecting tree distribution.
\begin{enumerate}
\item \label{itm:connecting_components} For any choice of trees $T_i\in\T_i$ for $i\in\{1,\ldots,k\}$, the multigraph $H_0\coloneqq (V, \bigdotcup_{i \in \{1, \dots, k\}}E[T_i])$ consists of one large connected component and potentially several isolated vertices.
The large connected component contains a walk with visits at $d_1,\ldots,d_k$ in this order that can be constructed efficiently from the trees $T_i$.
\item \label{itm:isolation_probability} If, in the above construction, the trees $T_i$ are sampled with marginals $\mu^i$, we have that for all $v\in V\setminus\{d_1,\ldots,d_k\}$,
\begin{equation*}
\Prob[v\text{ is isolated in }H_0] \leq \frac1\e\enspace.
\end{equation*}
\end{enumerate}
\end{lemma}

\begin{proof}
For \cref{itm:connecting_components} observe that each tree~$T_i$ is connected within itself by definition and, as it contains~$d_i$ and~$d_{i + 1}$, the union of all trees form one large connected component, while all other components must be isolated vertices.
Also, because each tree $T_i$ contains a $d_i$-$d_{i + 1}$ path, we may concatenate these paths to obtain the desired walk with visits at $d_1$, \ldots, $d_k$ in this order.

To prove~\cref{itm:isolation_probability}, we calculate the probability that a vertex~$v\in V\setminus\{d_1,\ldots,d_k\}$ is isolated in~$H_0$.
First, note that for any such vertex $v$ and any $i\in \{1,\ldots,k\}$, we have
\begin{equation*}
   \Prob\left[v \notin V\left[T_i\right]\right] = 1 - \sum_{T\in\T_i \colon v \in V[T]} \mu^i_T \enspace .
\end{equation*}
Thus, the probability that a vertex $v\in V\setminus\{d_1,\ldots,d_k\}$ is not contained in any tree $T_i$ for $i\in\{1,\ldots,k\}$, and hence is isolated in $H_0$, can be bounded as follows:
\begin{multline*}
   \Prob\!\left[v \notin \bigcup_{i = 1}^k V[T_i]\right]
   =    \prod_{i = 1}^{k} \Prob[v \notin V\left[T_i\right]]
   =    \prod_{i = 1}^{k} \left(1 - \sum_{T\in\T_i \colon v \in V[T]} \mu^i_T\right)\\
   \leq \exp\!\left(-\sum_{i=1}^k\sum_{T\in\T_i \colon v \in V[T]} \mu^i_T\right)
   =  \frac{1}{\e}\enspace,
\end{multline*}
where we used that $1-t\leq \exp(-t)$ for all $t\in\mathbb{R}$, and $\sum_{i = 1}^k \sum_{T \in \T_i \colon v \in V[T]} \mu^i_T = 1$ because $(\T_i,\mu^i)$ is a connecting tree distribution.
\end{proof}

Next, we bound the cost of the minimum-cost connector $F$ computed in \cref{algline:connector_F} of \cref{alg:ordered_tsp}.

\begin{lemma}\label{lem:connector_cost}
Let~$G = (V, E)$ be a graph,~$d_1, \dots, d_k \in V$, and let $T$ be a minimum spanning tree of~$G$.
\begin{enumerate}
\item \label{itm:cost_of_f_prime} For all $v\in V\setminus\{d_1\}$, let $e_v$ denote the unique edge outgoing of $v$ when orienting $T$ towards $d_1$.
For every graph $H_0$ on the vertex set $V$ with components that are---up to possibly the component containing~$d_1$---singleton vertices,
the minimum-cost edge set $F$ that connects $H_0$ satisfies
\begin{equation*}
c(F) \leq \sum_{v\text{ isolated in }H_0} c(e_v)\enspace.
\end{equation*}
\item \label{itm:cost_of_f}
Let $(\T_i, \mu^i)$ for $i\in\{1,\ldots,k\}$ be a connecting tree distribution.
If the trees $T_i\in\T_i$ are sampled with marginals $\mu^i$ and $H_0\coloneqq (V, \bigdotcup_{i \in \{1, \dots, k\}}E[T_i])$, we obtain
\begin{equation*}
\Exp[c(F)] \leq \frac1\e c(T)\enspace.
\end{equation*}
\end{enumerate}
\end{lemma}
\begin{proof}
In order to prove \cref{itm:cost_of_f_prime}, we construct a feasible connecting edge set $F'$ as the set of all edges~$e_v$ for which~$v$ is an isolated vertex.
Then $H_0 \cup F'$ is indeed connected, because each isolated vertex of~$H_0$ is connected to its predecessor in $T$ by an edge of~$F'$, hence inductively, the component of $H_0$ containing~$d_1$ can be reached along edges of $F'$.
As the minimum-cost connector~$F$ has cost at most $c(F')$, we have
\begin{equation*}
   c(F) \leq c(F') \leq \sum_{v\text{ isolated in }H_0} c(e_v)\enspace.
\end{equation*}
To prove \cref{itm:cost_of_f}, we note that in this case, $H_0$ consists of one large connected component and some isolated vertices by \cref{itm:connecting_components} of \cref{lem:tree_cost}.
Using \cref{itm:isolation_probability} of \cref{lem:tree_cost} on top of the above, we get
\begin{equation*}
   \Exp[c(F)] \leq \Exp[c(F')] = \sum_{v \in V\setminus\{d_1\}} \Prob[v\text{ isolated in }H_0] \cdot c(e_v) \leq \frac 1 \e \sum_{v \in V\setminus\{d_1\}} c(e_v) = \frac 1 \e c(E[T]) \enspace. \qedhere
\end{equation*}
\end{proof}

The cost of the $\odd(H)$-join $J$ constructed in \cref{algline:connector_F} of \cref{alg:ordered_tsp} can be bounded by $\frac12 c^\top x$ by \cref{item:join} of \cref{cor:tree_join_cost}.
Hence, to complete the analysis of \cref{alg:ordered_tsp}, it is left to show that from the Eulerian graph~$H\dotcup J$ constructed in \cref{algline:join_computation} of \cref{alg:ordered_tsp}, we can obtain an \OTSP solution of no larger cost.
We remark that such a step has also been used by \textcite{boeckenhauer_2013_improved}; we repeat it here explicitly and give a slightly different proof for completeness.
In the proof, we repeatedly use the operation of \emph{shortcutting} a vertex $v$ on a walk, which is the following:
If the predecessor and successor of $v$ on the walk are $u$ and $w$, respectively, we delete the edges $\{u,v\}$ and $\{v,w\}$ from the walk and add the direct edge $\{u,w\}$ instead.
It is clear that this operation results in a walk again; furthermore, by the triangle inequality, the costs of the walk do not increase under such operations.

\begin{lemma}\label{lem:spanning_cycle}
Let $G=(V,E)$ be complete graph with metric edge costs, and let $d_1,\ldots,d_k\in V$ be distinct.
Given an undirected connected Eulerian multigraph $M=(V,E_M)$ together with a closed walk in $M$ with visits at $d_1, \dots, d_k$ in this order, we can in polynomial time determine a spanning cycle $C$ in $G$ with visits at $d_1, \dots, d_k$ in this order of cost at most $c(E_M)$.
\end{lemma}

\begin{proof}
Let~$C$ be the given closed walk on which~$d_1, \dots, d_k$ appear in this order, delete $C$ from $M$ and partition the remaining Eulerian graph into a set $\mathcal{W}$ of closed walks.
Shortcut~$C$ to a cycle while maintaining visits at $d_1, \dots, d_k$ in this order.
This can, for example, be done by traversing~$C$ starting at~$d_1$, and shortcutting
\begin{enumerate*}
\item vertices that have already been visited, and
\item vertices~$d_i$ that are not yet to be visited due to the order constraint.
\end{enumerate*}
Afterwards, as long as~$\mathcal{W}$ is non-empty, pick a closed walk~$W$ from~$\mathcal{W}$ that intersects~$C$, and let~$v$ be a vertex in the intersection.
Traversing~$W$ starting from~$v$, shortcut~$W$ to a cycle by skipping, except for~$v$ itself, all vertices that are already contained in~$C$.
Then, merge~$W$ into~$C$ by first traversing~$C$ up to (and including)~$v$, then completely traversing~$W$ until (but not including)~$v$ before continuing on~$C$, thereby including only one visit at $v$ in the updated $C$.
It is immediate that~$C$ is still a cycle after any such operation, and the vertices \mbox{$d_1,\ldots,d_k$} still appear on~$C$ once and in this order.
By connectivity of $M$, this procedure only terminates once $\mathcal{W}$ is empty, and in that case, $C$ is a spanning cycle of $G$.
Also, all steps can be implemented to run in polynomial time.
Clearly, the final length of $C$ with respect to $c$ is at most $c\left(E_M\right)$ because $c$ is metric.
\end{proof}

From the above ingredients, we can readily prove \cref{thm:main_randomized}.

\begin{proof}[Proof of \cref{thm:main_randomized}]
The solution returned by \cref{alg:ordered_tsp} is a spanning cycle~$C$ in~$G$ obtained from~$H \dotcup J$ through \cref{lem:spanning_cycle}, hence it is feasible and of cost at most~$c(E[H \dotcup J])$.
Note that the required closed walk in $H\dotcup J$ with visits at $d_1$,~\ldots, $d_k$ in this order is guaranteed and can be constructed efficiently from the trees $T_i$ by \cref{itm:connecting_components} of \cref{lem:tree_cost}.
Furthermore, by \cref{itm:cost_of_f} of \cref{lem:connector_cost} and \cref{cor:tree_join_cost}, we know that $\Exp[c(F)]\leq \sfrac1\e\cdot c(T) \leq \sfrac 1 \e \cdot \costLP$, where~$T$ is a minimum-cost spanning tree.
In addition, \cref{cor:tree_join_cost} also implies that $c(E[J])\leq \frac12 \cdot c_\text{LP}$.
Last but not least, we can express the expected cost of each~$T_i$ as
\begin{equation*}
     \Exp[c(E[T_i])]
   = \sum_{T \in \T_i} \mu^i_T c(E[T])
   = \sum_{T \in \T_i} \mu^i_T c^\top\chi^{E[T]}
   = c^\top x^i \qquad \forall i \in \{1, \dots, k\}\enspace.
\end{equation*}
Thus, by summing over all constructed trees, we obtain
         $\sum_{i=1}^k \Exp[c(E[T_i])]
   =     \sum_{i = 1}^k c^{\top} x^i
   = \costLP$.
Together, this yields the proclaimed bound
\begin{align*}
   \Exp[c(C)]
    \leq   \Exp\left[c(E[H \dotcup J])\right]
    \leq      \left(\frac{3}{2} + \frac 1 \e\right) \cdot \costLP\enspace .
\end{align*}
It remains to note that \cref{alg:ordered_tsp} can be implemented to run in polynomial time.
To start with, an optimal solution of the \ref{eq:lp_relaxation} can be found in polynomial time because $\Pstroll$ admits a polynomial-time separation oracle through polynomially many calls to a minimum-cut algorithm.
Next, the decomposition in \cref{algline:decomposition} is obtained in polynomial time, finding an optimal edge set $F$ in \cref{algline:connector_F} can be implemented by Prim's algorithm, and the~$\odd(H)$-join is well-known to be computable in polynomial time.
Finally, also the computation of the cycle~$C$ in \cref{algline:H_cup_J} is polynomial due to \cref{lem:spanning_cycle}, concluding the proof.
\end{proof}

\subsection{Derandomizing \texorpdfstring{\cref{alg:ordered_tsp}}{Algorithm~\ref{alg:ordered_tsp}}}
\label{subsec:derandomization}

To complete a proof of our main result, \cref{thm:main}, we now show how to derandomize \cref{alg:ordered_tsp} using the method of conditional expectations, which results in the following proof.

\begin{proof}[Proof of \cref{thm:main}]
   By the construction of the solution $C$ in \cref{alg:ordered_tsp}, using \cref{itm:cost_of_f_prime} of \cref{lem:connector_cost} to bound the cost of $F$, and \cref{item:join} of \cref{cor:tree_join_cost} to bound the cost of $J$, we know that
   \begin{align}
      \nonumber
      c(C) &\leq \sum_{i=1}^k c(E[T_i]) + c(F) + c(E[J]) \\
      \label{eq:cost_bound}
      &\leq \underbrace{\sum_{i=1}^k c(E[T_i]) + \sum_{v\notin\bigcup_{i=1}^k V[T_i]} c(e_v) + \frac12\cdot c_\text{LP}}_{\eqqcolon g(T_1,\ldots,T_k)}\enspace,
   \end{align}
   where we recall that $e_v$, for $v\in V\setminus\{d_1\}$, is the unique outgoing edge at $v$ when orienting a minimum-cost spanning tree of $G$ towards $d_1$.
   For \cref{thm:main_randomized}, we showed that $\Exp[g(T_1,\ldots,T_k)]\leq (\sfrac32+\sfrac1\e)\cdot c_{\text{LP}}$.
   Following the method of conditional expectations, in order to derandomize the choices of the trees $T_i$ in \cref{algline:sample_T} of \cref{alg:ordered_tsp} while maintaining the upper bound on the solution cost, we sequentially choose trees~$S_i$ for $i\in\{1,\ldots, k\}$ such that
   \begin{equation}\label{eq:determininstic_choice}
   S_i = \argmin_{S\in\T_i}\ \Exp\left[g(T_1,\ldots,T_k) \;\middle|\; T_1 = S_1, \ldots, T_{i-1}=S_{i-1}, T_i = S\right]\enspace.
   \end{equation}
   Note that feasibility of the cycle $C$ and the bound of~\eqref{eq:cost_bound} on its cost are unaffected by fixing $T_i=S_i$.
   By definition of conditional expectation, we know that
   \begin{multline*}
      \Exp\left[g(T_1,\ldots,T_k) \;\middle|\; T_1 = S_1, \ldots, T_{i-1}=S_{i-1}\right] \\  = \sum_{S\in\T_i} \mu^i_S \cdot \Exp\left[g(T_1,\ldots,T_k) \;\middle|\; T_1 = S_1, \ldots, T_{i-1}=S_{i-1}, T_i=S\right]\enspace,
   \end{multline*}
   hence the sequence of conditional expectations $(\Exp[g(T_1,\ldots,T_k)\mid T_1 = S_1, \ldots, T_{i}=S_{i}])_{i\in\{1,\ldots,k\}}$ is non-increasing by the choice in~\eqref{eq:determininstic_choice}, because $\sum_{S\in\T_i}\mu_{S}=1$.
   Thus, it remains to observe that the conditional expectations in~\eqref{eq:determininstic_choice} can be computed.
   To this end, observe that
   \begin{multline*}
      \Exp\left[g(T_1,\ldots,T_k) \;\middle|\; T_1 = S_1, \ldots, T_{\ell}=S_{\ell}\right]\\
      = \sum_{i=1}^\ell c(E[S_i]) + \sum_{i=\ell+1}^k \Exp[c(E[T_i])] + \sum_{v\notin\bigcup_{i=1}^\ell V[S_i]} \mathllap{\Prob}\left[v\notin\bigcup_{i=\ell+1}^k V[T_i]\right] c(e_v) + \frac12\cdot c_{\text{LP}}\enspace,
   \end{multline*}
   and we can readily compute
   \begin{equation*}
      \Exp[c(E[T_i])] = \sum_{T\in\T_i} \mu^i_T c(E[T])
      \qquad\text{and}\qquad
      \Prob\left[v\notin\bigcup_{i=\ell+1}^k V[T_i]\right] = \prod_{i=\ell+1}^k (1-y^i_v)\enspace.\qedhere
   \end{equation*}
\end{proof}

\section{Extending to several independent total orders: Proving \texorpdfstring{\cref{thm:indep_orders}}{Theorem \ref{thm:indep_orders}}}
\label{sec:indep_orders}

In this section, we show how our approach can be extended to \TSPPC with a specific structure of precedence constraints that corresponds to having total orders on disjoint subsets $D_1,\ldots,D_\ell\subseteq V$ of the input graph~$G=(V,E)$.

As mentioned in the introduction, our approach is inherently tied to handle total orders---which is why, in the aforementioned setup, our solutions will not interleave vertices from different chains $D_j$, but rather treat the chains $D_j$ one after another.
Still, our approach allows to do better than simply constructing \OTSP solutions for all subinstances $(G, c, D_j)$ in a black-box way and concatenating them with appropriate shortcutting.
The latter would lead to an immediate $(\sfrac32+\sfrac1e)\ell$-approximate solution by using \cref{alg:ordered_tsp} on each subinstance.
Instead, we observe that after solving the \ref{eq:lp_relaxation} and sampling trees for each subinstance as in \cref{alg:ordered_tsp}, we may join all edges obtained this way and only \emph{once} need to connect remaining singletons and do parity correction.
This leads to \cref{alg:indep_orders} as stated below.

Note that, deviating from the above outline, \cref{alg:indep_orders} starts by guessing a root node $d_0$ among the minimal nodes in all sets $D_j$ with respect to $\prec$; this node is used as a common anchor of the given partial orders and results in connectivity of the multigraph containing all sampled trees.
To be able to compare the obtained solution to an optimal solution, we need $d_0$ to be, among the minimal nodes in all sets $D_j$, the first one to appear on an optimal solution.
We remark that for one $j\in\{1,\ldots,\ell\}$, we already have $d_0\in D_j$.
For the sake of uniform notation, we still add a copy of $d_0$ to $D_j$ in \cref{algline:add_d0} of \cref{alg:indep_orders}.

\begin{algorithm2e}[!ht]
   \KwIn{\TSPPC instance $(G, c, \prec)$ on graph $G=(V,E)$, where $\prec$ precisely induces total orders on disjoint subsets $D_1,\ldots,D_\ell\subseteq V$.}
   \medskip

   \DontPrintSemicolon
   Guess a root node $d_0$ among the minimal nodes in $D_i$ with respect to $\prec$.\;
   \ForEach{$j \in \{1, \dots, \ell\}$}
   {
       Compute an optimal solution~$(x^{ji}, y^{ji})_{i\in\{0, 1,\ldots,|D_j|\}}$ to the \ref{eq:lp_relaxation} for the \OTSP instance $(G, c, \{d_0\}\dotcup D_j)$ with an order given by $\prec$ extended by $d_0\prec D_j$.\label{algline:add_d0}\;
       \ForEach{$i \in \{0, 1, \dots, |D_j|\}$}
       {
          Apply \cref{lem:decomposition} to decompose~$(x^{ji},y^{ji})$ into trees~$\T_{ji}$ with weights $\mu^{ji}$. \label{algline:general_decomposition}\;
          Sample one tree~$T_{ji}$ from $\T_{ji}$ with marginals given by~$\mu^{ji}$.\label{algline:general_sample_T}\;
       }
   }
   Compute a minimum-cost edge set $F\subseteq E$ such that the multigraph
   $$H \coloneqq \left(V, F\cup\bigdotcup_{j=1}^\ell\bigdotcup_{i \in \{1, \dots, |D_j|\}} E[T_{ji}]\right)$$
   is connected.\label{algline:general_connector_F}\;
   Let~$Q = \operatorname{odd}(H)$ and compute a minimum cost~$Q$-join~$J$ in~$G$. \label{algline:general_join_computation}\;
   \Return Shortest spanning cycle $C$ in $G$ (over all guesses of $d_0$) that visits $d_0$, $D_1\setminus\{d_0\}$, \ldots, $D_\ell\setminus\{d_0\}$ in this order (while respecting $\prec$ in each $D_i$) and is obtained from~$H\dotcup J$ through \cref{lem:spanning_cycle}.\label{algline:general_H_cup_J}\;
   \caption{Approximating a special case of \TSPPC.}
   \label{alg:indep_orders}
\end{algorithm2e}

We show that this algorithm gives the guarantee claimed by \cref{thm:indep_orders} in expectation, and that it can be derandomized using the method of conditional expectations in a way analogous to the derandomization of \cref{alg:ordered_tsp}.

\begin{proof}[Proof of \cref{thm:indep_orders}]
Let $\costOPT$ denote the cost of an optimal solution of the given \TSPPC instance.
For every~$j\in\{1,\ldots,\ell\}$, note that the value $\costLP^j$ of the optimal solution $(x^{ji}, y^{ji})_{i\in\{1,\ldots,|D_j|\}}$ to the \OTSP instance $(G, c, \{d_0\}\dotcup D_j)$ generated in \cref{algline:add_d0} of \cref{alg:indep_orders} satisfies $\costLP^j\leq\costOPT$.
For every $j\in\{1,\ldots,\ell\}$, denote
$$
H_j \coloneqq\left(V,\bigdotcup_{i=0}^{|D_j|} E[T_{ji}]\right)\enspace.
$$
Every such graph is composed of trees from a connecting tree distribution.
Hence, by \cref{itm:connecting_components} of \cref{lem:tree_cost}, $H_j$ consists of a large connected component that contains a walk with visits at $d_0$ and all vertices of $D_j$ in the desired order, and potentially isolated vertices.
For all $v\notin\{d_0\}\cup D_j$, \cref{itm:isolation_probability} of \cref{lem:tree_cost} implies that
$$
\Prob[v\text{ is isolated in }H_j] \leq \frac1\e\enspace.
$$
Also, observe that
$$
\Exp[c(E[H_j])] = \sum_{i=0}^{|D_j|} \Exp[c(T_{ji})] = \sum_{i=0}^{|D_j|} \sum_{T\in\T_{ji}}\mu^{ji}_T c(E[T]) = \sum_{i=0}^{|D_j|} c^\top x^{ji} = \costLP^j \leq \costOPT\enspace.
$$
Consequently, the multigraph $H_0\coloneqq\bigdotcup_{j\in\{1,\ldots,\ell\}} H_j$ has total edge cost at most $\ell\cdot\costOPT$.
Furthermore, $H_0$ has one large connected component that contains a walk with visits at $d_0$, $D_1\setminus\{d_0\}$, \ldots, $D_j\setminus\{d_0\}$ in this order (obtained by concatenating the walks obtained in the graphs $H_j$ above), i.e., a walk that respects $\prec$.
Also, because the graphs $H_1,\ldots,H_\ell$ are independent,
$$
\Prob[v\text{ is isolated in } H_0] = \prod_{j=1}^\ell \Prob[v\text{ is isolated in } H_j] \leq \frac1{\e^\ell}\enspace.
$$
Hence, by \cref{itm:cost_of_f_prime} of \cref{lem:connector_cost}, the cost of the minimum-cost edge set $F$ connecting $H_0$, as constructed in \cref{algline:general_connector_F} of \cref{alg:indep_orders}, can be bounded as follows:
$$
\Exp[c(F)] \leq \sum_{v\text{ isolated in }H_0} \Prob[v \text{ is isolated in } H_0]\cdot c(e_v) \leq \frac1{\e^\ell} \cdot c(T) \leq \frac1{\e^\ell}\cdot\costOPT\enspace.
$$
Here, we used that for any $j\in\{1,\ldots,\ell\}$, we have $c(T)\leq\costLP^j$ by \cref{item:tree} of \cref{cor:tree_join_cost}, and $\costLP^j\leq\costOPT$ as mentioned above.
Similarly, by \cref{item:join} of \cref{cor:tree_join_cost}, we know that the cost of a cheapest $\odd(H)$-join~$J$ in the multigraph $H=H_0\cup F$ can be bounded by $c(E[J])\leq \frac12\cdot\costLP^j$ for any $j\in\{1,\ldots,\ell\}$, hence $c(E[J])\leq \frac12\cdot\costOPT$.

Altogether, we obtain a connected Eulerian multigraph $H\dotcup J$ together with a walk that has visits at $d_0$, $D_1\setminus\{d_0\}$, \ldots, $D_j\setminus\{d_0\}$ in the order given by $\prec$, and
$$
\Exp[c(E[H\dotcup J])] \leq \left(\ell + \frac12+\frac1{\e^\ell}\right)\cdot\costOPT\enspace.
$$
Thus, by \cref{lem:spanning_cycle}, we can efficiently find a cycle with visits at $d_0$, $D_1\setminus\{d_0\}$, \ldots, $D_j\setminus\{d_0\}$ in the order given by $\prec$ of at most the above expected cost.

Finally, to derandomize the random selection of trees $T_{ji}$ in \cref{alg:indep_orders}, we observe that the present randomized analysis relies on a bound of the form
\begin{equation*}
c(C) \leq \sum_{j=1}^\ell \sum_{i=0}^{|D_j|} c(T_{ji}) + \sum_{v\notin\bigdotcup_{j\in\{1,\ldots,\ell\}}\bigdotcup_{i \in \{1, \dots, |D_j|\}}V[T_{ji}]} c(e_v) + \frac12\cdot\costOPT \enspace.
\end{equation*}
The conditional expectations of this bound with respect to fixing any subset of the trees $T_{ji}$ can be readily computed.
Thus, the derandomization works analogously to \cref{alg:ordered_tsp} by the method of conditional expectations, in each iteration fixing one of the~$T_{ji}$.
To complete the proof of \cref{thm:main_randomized}, we observe that all steps of \cref{alg:indep_orders} can be implemented to run in polynomial time.
\end{proof}

\begin{remark}
We remark that the analysis of \cref{alg:indep_orders} above is with respect to the actual cost $\costOPT$ of an optimal \TSPPC solution.
Alternatively, after guessing a root node $d_0$, one could also write an LP relaxation generalizing the \ref{eq:lp_relaxation} by introducing independent copies of the variables for each chain $\{d_0\}\cup D_j$ and minimizing the cost of a point $x\in\PHK(G)$ that dominates the edge usage within each of the copies.
For the ease of presentation, though, we decided to present the above analysis only.
\end{remark} %
\section{Proof of \texorpdfstring{\cref{lem:decomposition}}{Lemma~\ref{lem:decomposition}}}
\label{sec:decomposition_proof}

As mentioned earlier, we derive \cref{lem:decomposition} from a closely related result used by \textcite[Lemma~4.2]{blauth_2023_improved}.
We restate their result here in a slightly simplified form that follows immediately from the original formulation.

\begin{lemma}[{\cite[Lemma~4.2]{blauth_2023_improved}}]\label{lem:general_decomposition}
Let $G=(V,E)$ be a graph with $r\in V$, let $(x, y)\in\mathbb{R}_{\geq 0}^E\times\mathbb{R}_{\geq 0}^V$ be feasible for the system
\begin{equation}\label{eq:decomposition_system}
\begin{array}{rcll}
x(\delta(v)) & = & 2y_v & \forall v\in V \\
x(\delta(S)) & \geq & 2y_v & \forall S\subseteq V\setminus\{r\}, v\in S\\
y_r & = & 1\enspace,
\end{array}
\end{equation}
and assume that there is a vertex $u\in V\setminus\{r\}$ such that $y_u=1$ and $e_0=\{u, r\}$ satisfies $x_{e_0}\geq 1$.
We can in polynomial time construct a set $\T$ of trees that all contain the vertices $r$ and $u$, and weights $\mu \in [0,1]^\T$ with $\sum_{T\in \T}\mu_T = 1$ and the following properties:
\begin{enumerate}

\item\label{item:solution_partitioning} The point $x\in\mathbb{R}_{\geq 0}^E$ is a conic combination of the trees in $\T$ with weights $\mu$ and the edge $e_0$, i.e.,
$$
x = \sum_{T\in \T}\mu_T\chi^{E[T]} + \chi^{e_0} \enspace.
$$

\item\label{item:incident_trees} For every $v\in V \setminus U$,
$$
\sum_{T\in \T\colon v\in V[T]} \mu_T = y_v \enspace.
$$
\end{enumerate}
\end{lemma}

The proof of \cref{lem:general_decomposition} relies on the well-known \emph{splitting-off} technique (see, e.g., \cite{lovasz_1976_connectivity,mader_1978_reduction,frank_1992_on}) applied in the graph $G$ with weights $x$.
Indeed, the constraints in the system \eqref{eq:decomposition_system} can be interpreted as $r$-$v$ connectivity requirements for all $v\in V\setminus\{r\}$, hence splitting-off allows to remove a vertex from the graph while preserving the connectivity properties of the remaining graph.
An inductive construction of the desired family of trees is then achieved by reverting the splitting-off operations and extending trees appropriately.
For a complete proof, we refer to \textcite{blauth_2023_improved}.

To deduce \cref{lem:decomposition} from \cref{lem:general_decomposition}, we note that a point $(x,y)\in\Pstroll$ can be easily transformed into a point $(x', y')$ satisfying the assumptions of \cref{lem:general_decomposition} by adding one unit to $x_{\{s,t\}}$ and adjusting $y_s$ and $y_t$ accordingly.
Note that intuitively, this corresponds to closing an $s$-$t$ stroll to obtain a tour by adding a copy of the edge $\{s,t\}$.

\begin{proof}[Proof of \cref{lem:decomposition}]
Given $(x,y)\in \Pstroll$, we assume without loss of generality that $e_0\coloneqq\{s,t\}\in E$ and define $x' \coloneqq x + \chi^{\{s,t\}}$ and $y' = y + \frac12(\chi^{s}+\chi^t)$.
We claim that $(x', y')$ with $r=s$ and $u=t$ satisfy the assumptions of \cref{lem:general_decomposition}.
Indeed, $y'_s=y'_t=1$, and $x'_{e_0} = x_{e_0} + 1 \geq 1$.
Moreover, for $v\notin\{s,t\}$, we have $x'(\delta(v)) = x(\delta(v)) = 2y_v$; for $v\in\{s,t\}$, we have $x'(\delta(v)) = x(\delta(v))+1 = 2 = 2y'_v$, hence the degree constraints in~\eqref{eq:decomposition_system} are satisfied.
Finally, to verify that the cut constraints of~\eqref{eq:decomposition_system} are satisfied, too, let $S\subseteq V\setminus \{r\}$ and $v\in S$.
If $t\notin S$, then $x'(\delta(S)) = x(\delta(S))\geq 2y_v = 2 y'_v$ follows from the corresponding constraint of $\Pstroll$.
If otherwise $t\in S$, we know that $x(\delta(S))\geq 1$, hence
$$
x'(\delta(S)) = x(\delta(S)) + 1 \geq 2 \geq 2y'_v\enspace,
$$
where we use that $y'_v=y_v\leq 1$ is implied by the constraints of $\Pstroll$ for $v\in V\setminus\{s,t\}$ (see \cref{fn:stroll_coverage_variables}), and $y'_s=y'_t=1$.

Consequently, by applying \cref{lem:general_decomposition} to $(x',y')$, we obtain in polynomial time a set $\T$ of trees that all contain $s$ and $t$, and weights $\mu \in [0,1]^\T$ with $\sum_{T\in\T} \mu_T=1$ such that
$$
x + \chi^{e_0} = x' = \sum_{T\in \T}\mu_T\chi^{E[T]} + \chi^{e_0} \enspace,
$$
i.e., $x=\sum_{T\in \T}\mu_T\chi^{E[T]}$, and, for every $v\in V \setminus \{s,t\}$,
\begin{equation*}
\sum_{T\in \T\colon v\in V[T]} \mu_T = y'_v = y_v \enspace.\qedhere
\end{equation*}
\end{proof}

\paragraph{Acknowledgement.} We thank Tobias M\"{o}mke for bringing the Ordered Traveling Salesperson Problem to our attention during the Aussois Combinatorial Optimization Workshop 2024.
\begingroup
\setlength{\emergencystretch}{1em}
\printbibliography
\endgroup

\end{document}